\documentclass[conference]{IEEEtran}
\usepackage{mathptmx} % small math font
\usepackage{times}
\usepackage{xcolor}
\usepackage{soul}
\usepackage[utf8]{inputenc}
\usepackage[small]{caption}
\usepackage{enumerate}
\usepackage{amsmath}
\usepackage{amsfonts}
\usepackage{wasysym}
\usepackage{amssymb}
\usepackage{amsthm}
\usepackage[authoryear]{natbib}
\usepackage{graphicx}

\newcommand{\ignore}[1]{}
\newcommand*{\Scale}[2][4]{\scalebox{#1}{$#2$}}%
\newtheorem{theorem}{Theorem}
\newtheorem{lemma}{Lemma}
\newtheorem{coro}{Corollary}
\newtheorem{proposition}{Proposition}

\newtheorem{claim}{Claim}

\newtheorem{definition}{Definition}

\newcommand{\OPT}{\mbox{\rm OPT}}

\newcommand{\EE}{\ensuremath{\mathbb{E}}}
\newcommand{\head}[1]
{\markright{\hbox to 0pt{\vtop to 0pt{\hbox{}\vskip 3mm \hrule
                width  \textwidth \vss} \hss}{\sc #1}}}
\title{{\bf Election with Bribed Voter Uncertainty: Hardness and Approximation Algorithm}}
\author{Lin Chen$^1$, Lei Xu$^2$, Shouhuai Xu$^3$, Zhimin Gao$^1$, Weidong Shi$^1$\\ 
	$^1$ Department of Computer Science, University of Houston, TX, USA\\ 
	$^2$ Conduent Labs, NC, USA \\ 
	$^3$ Department of Computer Science, University of Texas at San Antonio, TX, USA\\
	chenlin198662@gmail.com, xuleimath@gmail.com, shouhuai.xu@utsa.edu, gao@kell.vin, larryshi@ymail.com}
\begin{document}
%\title{Electoral Control Under Voter Uncertainty -- Complexity and Algorithms}
\maketitle
%\titlenote{Produces the permission block, and copyright information}
%\subtitle{Extended Abstract}
%\subtitlenote{The full version of the author's guide is available as \texttt{acmart.pdf} document}

\begin{abstract}
Bribery in election (or computational social choice in general) is an important problem that has received a considerable amount of attention. 
In the classic bribery problem, the briber (or attacker) bribes some voters in attempting to make the briber's designated candidate win an election.
In this paper, we introduce a novel variant of the bribery problem, ``Election with \underline{B}ribed \underline{V}oter \underline{U}ncertainty'' or BVU for short, accommodating the uncertainty that the vote of a bribed voter may or may not be counted.
This uncertainty occurs either because a bribed voter may not cast its vote in fear of being caught, or because a bribed voter is indeed caught and therefore its vote is discarded. 
As a first step towards ultimately understanding and addressing this important problem,
we show that it does not admit any {\em multiplicative} $O(1)$-approximation algorithm 
modulo standard complexity assumptions.
We further show that there is an approximation algorithm that returns a solution with an {\em additive}-$\epsilon$ error in FPT time for any fixed $\epsilon$. 
\end{abstract}

%
% The code below should be generated by the tool atn% http://dl.acm.org/ccs.cfm

%\keywords{Voting, electoral control, fixed parameter tractability, approximation algorithm}

\maketitle

\section{Introduction}
In multiagent systems, election (or voting) is an important mechanism for collective decision-making. This importance has led to extensive investigations of various aspects of election. Indeed, the field of {\em computational social choice} investigates algorithmic and computational complexity aspects of this mechanism (see, e.g., the book by \cite{brandt2016handbook}).
In this paper, we focus on two important aspects of election that have received an extensive amount of attention but are still not fully understood: {\em uncertainty} and {\em bribery}.

\smallskip
\noindent\textbf{Uncertainty.}
Most studies in election investigated {\em deterministic} models and did not consider {\em uncertainty}, which is however often encountered in real-world scenarios.
There are two exceptions.
One exception is the investigation of uncertainty from the perspective of the {\em possible winner}. In this perspective, the input is {\em incomplete} and the problem is to determine if it is possible to extend the incomplete input to make a designated candidate win or lose. The uncertainty can be incurred by voters' incomplete preference lists, as shown by \cite{konczak2005voting,xia2011determining,betzler2010towards,baumeister2012taking,betzler2009multivariate}. The uncertainty can also be incurred by an incomplete set of candidates (e.g., additional candidates may be added), as shown by \cite{chevaleyre2010possible,xia2011possible,baumeister2011computational}.
The other exception is the investigation of uncertainty incurred by complete but {\em probabilistic} inputs. For example, \cite{wojtas2012possible} introduced an election model in which voters {or candidates} may have some {\em probabilities of no-show},
either because the communication network is not reliable or because voters inherently behave as such.

\smallskip
\noindent\textbf{Bribery.}
\cite{faliszewski2009hard} introduced the {\em bribery} problem in which a briber (or attacker) attempts to make a designated candidate win by paying a (monetary) bribe to some voters. Once bribed, a voter will vote for the candidate designated by the attacker. This problem has received a considerable amount of attention; see, e.g., \citet{lin2010complexity,brelsford2008approximability,xia2012computing,faliszewski2015complexity,faliszewski2011multimode,faliszewski2009llull,parkes2012complexity}. Most studies in this context consider {\em deterministic} models,
but researchers have started investigating the issue of {\em uncertainty} in this context as well. For example, \cite{erdelyi2014bribery} considered the bribery problem with uncertain voting rules; \cite{mattei2015complexity} considered the bribery problem
with uncertain information, \cite{erdelyi2009complexity} considered uncertainty in the lobbying problem, which is related to, but different from, the bribery problem.

\smallskip
\noindent\textbf{New problem: Election with \underline{B}ribed \underline{V}oter \underline{U}ncertainty (BVU).}
We observe that in the context of the bribery problem, there is an inherent uncertainty that has not been considered in the literature:
The vote of a bribed voter may or may not be counted,
either because a bribed voter may choose not to cast its vote in fear of being caught,
or because a bribed voter is indeed caught and therefore its vote is discarded.
In this setting, each voter is associated with a price of being bribed as well as a probability that its vote is not counted upon taking a bribe. The goal of the attacker is to bribe a subset of voters such that the total bribing cost does not exceed a given budget, while the probability that a designated candidate wins the election is maximized.

The importance of understanding bribed voter uncertainty cannot be overestimated. This is because, even with the proliferation of anonymous and unregulated cryptocurrencies (e.g., Bitcoin) that are deemed as ideal for bribery purposes, there is still a possibility that a bribe-taking voter is detected (see \cite{DBLP:journals/corr/abs-1708-04748}). In the United States, telling a voter whom to vote for is one type of voting fraud and may cause the votes to be discarded (see\citet{govfile}), as attested by the case that the Wetumpka City Council District 2 election was switched after 8 ballots were ruled (by a judge) to be thrown out (see\citet{overturn}).

\subsection{Our Contributions}

In this paper we make three main conributions. First, we introduce and initiate the study of the BVU problem, which captures a new form of uncertainty in bribery.

Second, we characterize the hardness of the BVU problem and show that the newly captured uncertainty completely changes the complexity
of the bribery problem as follows. In the absence of uncertainty, the bribery problem can be solved by a simple greedy algorithm (as shown by \cite{faliszewski2009hard}).
In the presence of uncertainty, assuming $P\neq NP$, there is no $O(1)$-approximation algorithm even if there are only two candidates;
assuming $W[1]\neq \textit{FPT}$, {there is no} $O(1)$-approximation algorithm that runs in FPT time parameterized by $r$,
which is the difference between the number of votes received by the winner and the number of votes received by the designated candidate in the absence of bribery.

Third, despite the strong hardness result mentioned above,
we show the existence of an additive $\epsilon$-approximation FPT algorithm when the number of candidates is a constant. This means that for an arbitrary small $\epsilon>0$, there is an algorithm that runs in FPT-time (parameterized by the parameter $r$ mentioned above) and returns an approximate solution with an objective value that is at most $\epsilon$ smaller than the optimal objective value. This result relies on a reduction from the BVU problem to a new variant of the knapsack problem (involving a stochastic objective and multiple cardinality constraints) and an approximation algorithm for this new variant of the knapsack problem (while leveraging dynamic programming and a non-trivial application of Berry-Essen's Theorem). Both the proof technique and the new variant of the knapsack problem may be of independent interest.

%All of the omitted proofs can be found in the full version of the present paper {\color{red}\cite{AAAI2019-full-version}}.

\section{Problem Statement and Preliminaries}
\label{sec:model}

%\smallskip\noindent\textbf{The basic election model.}
\noindent\textbf{\bf The basic election model.}
In the basic election model, there are a set of $m$ candidates $\mathcal{C}=\{c_1,c_2,\ldots, c_m\}$ and a set of $n$ voters $\mathcal{V}=\{v_1,v_2,\ldots,v_n\}$. Each voter has a preference over the candidates. There is a voting {\em rule} according to which a winner is selected. In this paper we focus on the {\em plurality rule} { with a single winner}, namely that every voter votes for its most preferred  candidate and the winner is the candidate that receives the highest number of votes.

\smallskip	
%\smallskip\noindent\textbf{The classic bribery problem in the basic election model.} 
\noindent\textbf{\bf The classic bribery problem in the basic election model.} 
A voter may be bribed to deviate from its own preference. Suppose each voter $v_i$ has a price $q_i$. If $v_i$ takes a bribe of amount $q_i$ from the briber (or attacker), then $v_i$ will vote, regardless of $v_i$'s own preference, for the {\em designated candidate} of the brier (i.e., the candidate preferred by the briber). The briber has a total bribe budget $Q$. The goal of the briber is to make the designated candidate win the election. The bribery problem has been extensively investigated in the literature; see, for example, \citet{faliszewski2009hard,lin2010complexity,brelsford2008approximability,xia2012computing,faliszewski2015complexity,parkes2012complexity}.
%say, $c_{m}$, win.

\smallskip
%\smallskip\noindent\textbf{New problem: BVU (Election with Bribed Voter Uncertainty).} 
\noindent\textbf{\bf BVU (Election with Bribed Voter Uncertainty): A new problem.} 
As discussed before, we introduce and study a novel variant of the classic bribery problem.
Suppose voter $v_i$ takes a bribe of amount $q_i$ from the briber.
With probability $p_i\in [0,1]$, which is independent of anything else, the vote of $v_i$ goes to the designated candidate and is counted;
with probability $1-p_i$, the vote of $v_i$ is {\em not} counted (for the two reasons mentioned above), that is, {\em no} candidate will receive the vote from $v_i$.
Without loss of generality, let $c_1$ be the winner when there is no bribery
and $c_m$ be the briber's designated candidate. Let $V_j$ be the subset of voters that vote for candidate $c_j$ in the absence of bribery, then $|V_1|>|V_j|$ for any $j>1$. Moreover, let $r=|V_1|-|V_{m}|$, namely the difference between the number of votes received by the winner $c_1$ and the number of votes received by the designated candidate $c_{m}$ in the absence of bribery. The BVU problem asks for a subset of voters in $V\setminus V_m$ whose total price is bounded by $Q$ such that if they are bribed, the probability that the designated candidate $c_m$ wins is maximized, while noting that the voters in $V_m$ do not need to be bribed because they already vote for $c_m$. More precisely, the BVU problem is formalized as follows,

\begin{center}
	\fbox{\begin{minipage}{0.47\textwidth}	
			\textbf{The (Plurality-)BVU Problem}	
			
			Input: 
			A set of $m$ candidates $\mathcal{C}=\{c_1,c_2,\ldots, c_m\}$, {where $c_1$ is the winner and $c_{m}$ is the designated candidate in the absence of bribery}; a set of $n$ voters $\mathcal{V}=\{v_1,v_2,\cdots,v_n\}$ with $\mathcal{V}=\cup_{j=1}^m V_j$, where $V_j$ is the subset of voters that vote for $c_j$ in the absence of bribery; a positive integer $r=|V_1|-|V_{m}|$; the briber's budget $Q$; each $v_i$ is associated with a price $q_i$ for bribe and a probability $p_i$ with which the vote of the {\em bribed} voter $v_i$ goes to the designated candidate $c_m$ and is counted (i.e., $1-p_i$ is the probability that the vote of the {\em bribed} $v_i$ is not counted)
			
			Output: Find a set of indices $I^*\subseteq \{1,2,\cdots,n\}$ such that
			\begin{itemize}
				\item $\sum_{i\in I^*} q_i\le Q$, and
				\item the probability that the designated candidate $c_{m}$ wins is maximized by bribing voters in $V'=\{v_i\in V\setminus V_m|i\in I^*\}$ 
			\end{itemize}
	\end{minipage}}
\end{center}

%\section{Preliminaries}
%\label{sec:preliminaries}

%We will use the following inequalities to prove our results.

\paragraph{Preliminaries.}
Let $Z$ be a random variable taking non-negative values. The Markov's inequality (see, for example, \citet{stein2009real}) says the following:
For any $a> 0$, it holds that
\begin{equation}
\label{eq:markov-inequality}
\Pr(Z\ge a)\le \frac{\EE(Z)}{a}.
\end{equation}

%\smallskip\noindent\textbf{Chernoff bound.}\citet{mitzenmacher2017probability}
%Suppose $Z_1, Z_2, \cdots, Z_n$ are independent random variables taking values in $\{0, 1\}$. Let $Z$ denote their sum and let $\mu = \EE(Z)$ denote the sum's expected value. Then for any $\delta\in (0,1)$,
%$${\Pr(X\leq (1-\delta )\mu )\leq e^{-{\frac {\delta ^{2}\mu }{2}}}.}$$
% $$\Pr(X\geq (1+\delta )\mu )\leq e^{-{\frac {\delta ^{2}\mu }{3}}}$$

\ignore{
	\smallskip\noindent\textbf{Berry-Essen's theorem} (\cite{berry1941accuracy}).
	Let $Z_1, Z_2,\ldots,Z_n$ be independent random variables with $\EE(Z_i) = 0$, $\EE(Z_i^2) = \sigma_i^2 > 0$, and $\EE(|Z_i|^3) = \rho_i <\infty$. Let
	$$S_{n}={Z_{1}+Z_{2}+\ldots +Z_{n} \over {\sqrt  {\sigma _{1}^{2}+\sigma _{2}^{2}+\ldots +\sigma _{n}^{2}}}}.$$
	Then, it holds that
	\begin{equation}
	\label{eq:Berry-Essen-theorem}
	\sup _{{x\in {\mathbb  R}}}\left|F_{n}(x)-\Phi (x)\right|\leq C_{0}\cdot \psi _{0},
	\end{equation}
	where $C_0$ is a universal constant, $F_n(x)$ is the cumulative distribution function of $S_n$, $\Phi(x)$ is the standard normal distribution $\mathcal{N}(0,1)$, and
	$$\psi _{0}={\Big (}{\textstyle \sum \limits _{{i=1}}^{n}\sigma _{i}^{2}}{\Big )}^{{-3/2}}\cdot \sum \limits _{{i=1}}^{n}\rho _{i}.$$
}

\begin{theorem}[Berry-Essen theorem; see \cite{berry1941accuracy}]
	\label{theorem:Berry-Essen}
	Let $Z_1, Z_2,\ldots,Z_n$ be independent random variables with $\EE(Z_i) = 0$, $\EE(Z_i^2) = \sigma_i^2 > 0$, and $\EE(|Z_i|^3) = \rho_i <\infty$. Let
	$$S_{n}={Z_{1}+Z_{2}+\ldots +Z_{n} \over {\sqrt  {\sigma _{1}^{2}+\sigma _{2}^{2}+\ldots +\sigma _{n}^{2}}}}.$$
	Then, it holds that
	\begin{equation}
	\label{eq:Berry-Essen-theorem}
	\sup _{{x\in {\mathbb  R}}}\left|F_{n}(x)-\Phi (x)\right|\leq C_{0}\cdot \psi _{0},
	\end{equation}
	where $C_0$ is a universal constant, $F_n(x)$ is the cumulative distribution function of $S_n$, $\Phi(x)$ is the standard normal distribution $\mathcal{N}(0,1)$, and
	$$\psi _{0}={\Big (}{\textstyle \sum \limits _{{i=1}}^{n}\sigma _{i}^{2}}{\Big )}^{{-3/2}}\cdot \sum \limits _{{i=1}}^{n}\rho _{i}.$$
\end{theorem}

The following {\bf Proposition \ref{lemma:folklore}} is a folklore.

\begin{proposition}\label{lemma:folklore}
	Let $Y_1,Y_2,Z_1,Z_2$ be independent random variables taking values in $\mathbb{Z}_{\ge 0}$ (i.e., {non-negative integers}) such that for any integer $0\le h\le N$ the following hold:
	$$\Pr(Y_1\ge h)\ge (1-\delta) \Pr(Z_1\ge h),$$
	$$\Pr(Y_2\ge h)\ge (1-\delta)\Pr(Z_2\ge h).$$
	Then, for any $0\le \ell\le N$, we have:
	$$\Pr(Y_1+Y_2\ge \ell)\ge (1-\delta)^2\Pr(Z_1+Z_2\ge \ell).$$
\end{proposition}

\begin{proof}
	For any integer $0\le \ell\le N$, we have
	\begin{eqnarray*}
		&&\Pr(Y_1+Y_2\ge \ell)\\
		&=&\sum_{j=0}^{\ell-1} \Pr(Y_1=j)\Pr(Y_2\ge \ell-j)+\Pr(Y_1\ge \ell)\\
		&\ge &(1-\delta)[\sum_{j=0}^{\ell-1} \Pr(Y_1=j)\Pr(Z_2\ge \ell-j)+\Pr(Y_1\ge \ell)]\\
		&\ge & (1-\delta)\Pr(Y_1+Z_2\ge \ell).
	\end{eqnarray*}
	Similarly, we can prove that
	\begin{eqnarray*}
		\Pr(Y_1+Z_2\ge \ell)\ge (1-\delta)\Pr(Z_1+Z_2\ge \ell).
	\end{eqnarray*}
	Hence, $\Pr(Y_1+Y_2\ge \ell)\ge (1-\delta)^2\Pr(Z_1+Z_2\ge \ell).$	
\end{proof}

{\bf Proposition \ref{lemma:folklore}} can be re-written additively as follows.
\begin{coro}\label{coro:folklore}
	Let $Y_1,Y_2,Z_1,Z_2$ be independent random variables taking values in $\mathbb{Z}_{\ge 0}$ such that for any integer $0\le h\le \ell$, the following hold:
	$$\Pr(Y_1\ge h)\ge \Pr(Z_1\ge h)-\delta,$$
	$$\Pr(Y_2\ge h)\ge \Pr(Z_2\ge h)-\delta.$$
	Then, we have:
	$$\Pr(Y_1+Y_2\ge \ell)\ge \Pr(Z_1+Z_2\ge \ell)-2\delta.$$
\end{coro}

By iteratively applying {\bf Corollary \ref{coro:folklore}} to a sequence of independent random variables, we obtain the following corollary that will be used later.
% Corrollary \ref{coro:folklore-sequence}. %{ which will be used later to decompose a problem into }:

\begin{coro}\label{coro:folklore-sequence}
	Let $Y_j,Z_j$, $1\le j\le m$, be $2m$ independent random variables taking values in $\mathbb{Z}_{\ge 0}$ such that for any integer $0\le h\le \ell$ and $1\le j\le m$, the following holds:
	$$\Pr(Y_j\ge h)\ge (1-\delta)\Pr(Z_j\ge h)-\delta$$
	%	$$\Pr(Y_2\ge h)\ge (1-\delta)\Pr(Z_2\ge h)-\delta.$$
	Then, we have:
	$$\Pr\left(\sum_{j=1}^mY_j\ge \ell\right)\ge (1-\delta)^m\Pr\left(\sum_{j=1}^mZ_j\ge \ell\right)-m\delta.$$
\end{coro}

%{\smallskip\noindent\textbf{Gordon's inequality}. }
%We use the following inequality by Gordon\citet{gordon1941values} to bound the tail of the standard normal distribution:
%\begin{eqnarray}\label{eq:normal-tail}
%\frac{e^{{-t^2}/{2}}}{\sqrt{2\pi}}\cdot \frac{1}{t+1/t}\le \int_{t}^{+\infty}\frac{1}{\sqrt{2\pi}} e^{-x^2/2}dx\le \frac{e^{{-t^2}/{2}}}{\sqrt{2\pi}}\cdot \frac{1}{t}, \quad \forall t>0
%\end{eqnarray}
%There are various improved bounds and we refer the reader to a nice technical report\citet{duembgen2010bounding} which gives a survey. For this paper, the bound by Gordon suffices.

%\section{Inapproximability of the BU Problem}
\section{Hardness of the BVU Problem}
\label{sec:inapproximability-results}
We show the hardness of the BVU problem for $m=2$. By introducing {\em dummy} voters whose prices are higher than the briber's budget $Q$ (i.e., they cannot be bribed), {the hardness result immediately applies to the case of an arbitrary $m>2$}.

\subsection{Hardness Result}
The goal of this subsection is to prove the following.
\begin{theorem}[Main hardness result]
	\label{coro:pec-c-u}
	Assuming $W[1]\neq FPT$, there does not exist an $O(1)$-approximation algorithm for BVU problem that runs in FPT time parameterized by $r$, even if $m=2$.
	Moreover, assuming $P\neq NP$, there does not exist an $O(1)$-approximation algorithm for the BVU problem that runs in polynomial time if $r$ is part of the input, even if $m=2$. 
\end{theorem}

In order to prove {\bf Theorem~\ref{coro:pec-c-u}}, we leverage the equivalence between the BVU problem with $m=2$ and the following
Knapsack with Uncertainty (KU) problem.

\begin{center}
	\fbox{\begin{minipage}{0.47\textwidth}
			\textbf{Knapsack with Uncertainty (KU)}		
			
			Input: A knapsack of capacity $Q$; a set of $n'$ items, with each item associated with a size $q_i$ and a profit $P_i$, which is an independent random variable such that $\Pr(P_i=1)=p_i$ and $\Pr(P_i=0)=1-p_i$; a positive integer $r$.
			
			Output: Find a set of indices $I^*\subseteq \{1,2,\cdots,n\}$ such that
			\begin{itemize}
				\item $\sum_{i\in I^*} q_i\le Q$, and
				\item $\Pr\left(\sum_{i\in I^*} P_i\ge r+1-|I^*|\right)$ is maximized.
			\end{itemize}
	\end{minipage}}
\end{center}

\begin{lemma}
	\label{lemma:BU-KU-equivalence}
	The BVU problem with $m=2$ is equivalent to the KU problem. 
\end{lemma}

\begin{proof}[{\bf Proof of Lemma \ref{lemma:BU-KU-equivalence}}]
	Consider the BVU problem with $m=2$. Recall that $c_1$ is the winner in the absence of bribery, $c_2$ is the designated candidate, $r=|V_1|-|V_2|$, and the problem is to bribe a set $V'=\{v_i\in V_1|i\in I^*\}$ of voters so that the probability $c_2$ wins is maximized.

	Consider the number of votes received by candidates $c_1$ and $c_2$ {\em after} the briber bribes the voters in $V'$. 
	%Obviously there is no need to bribe voters in $V_2$, hence $V'\subseteq V_1$. 
	For a bribed voter $v_i\in V'$, there are two possibilities: 
	\begin{itemize}
		\item The vote of $v_i$ is counted, meaning the number of votes received by candidate $c_1$ decreases by 1 and the number of votes received by candidate $c_{2}$ increases by 1.
		\item The vote of $v_i$ is not counted, meaning the number of votes received by $c_1$ decreases by 1 but the number of votes received by $c_{2}$ remains the same. 
	\end{itemize}
	This means that the votes received by candidate $c_1$ decreases to $|V_1|-|V'|$. Hence, for $c_2$ to win, it needs at least $|V_1|-|V'|+1$ votes. Given that $c_2$ originally receives $|V_2|$ votes, at least $|V_1|-|V'|-|V_2|+1=r-|I^*|+1$ votes from the bribed voters are counted. Let $X_i$ be a binary random variable indicating whether the vote of $v_i$ is counted, then $\Pr(X_i=1)=p_i$ and $\Pr(X_i=0)=1-p_i$. The probability that at least $r-|I^*|+1$ votes of the bribed voters are counted is $\Pr\left(\sum_{i\in I^*} X_i\ge r+1-|I^*|\right)$. That is, the BVU problem with $m=2$ essentially asks for an index set $I^*$ such that $\sum_{i\in I^*} q_i\le Q$ and $\Pr\left(\sum_{i\in I^*} X_i\ge r+1-|I^*|\right)$ is maximized. This is exactly the KU problem.
	
	%This completes the proof of {\bf Lemma \ref{lemma:BU-KU-equivalence}}.
\end{proof}

In order to prove {\bf Theorem~\ref{coro:pec-c-u}}, we also need:

\begin{theorem}%[{\color{blue}one form of hardness of the KU problem}]
	\label{thm:hardness}
	Assuming $W[1]\neq FPT$, there does not exist an $O(1)$-approximation algorithm for the KU problem that runs in FPT time parameterized by $r$.
\end{theorem}

\begin{proof}[{\bf Proof of Theorem~\ref{thm:hardness}}]
	%In order to prove {\bf Theorem~\ref{thm:hardness}}, 
	We leverage the $d$-sum problem, which is known to be $W[1]$-hard (see \cite{downey1992fixed}), and show a reduction from the $d$-sum problem to the KU problem. We first review the $d$-sum problem.
	
	\begin{center}
		\fbox{\begin{minipage}{0.47\textwidth}
				\textbf{The $d$-sum Problem}
				
				Input: $s$ positive integer $x_1,x_2,\cdots,x_s$ and an integer $t$.
				
				Output: Decide whether or not there exists a subset $E\subseteq \{x_1,x_2,\cdots,x_s\}$ of $|E|=d$ elements such that $\sum_{i:x_i\in E} x_i=t$.
				
		\end{minipage}}
	\end{center}
	
	The rest is to show the following reduction: If there is an $\alpha$-approximation algorithm that solves the KU problem in $f(r)n^{O(1)}$ time for some computable function $f$ and some constant $\alpha$, then this algorithm can be used to solve the $d$-sum problem in $f(d)m^{O(1)}$ time.
	This contradicts the $W[1]$-hardness result of the $d$-sum problem mentioned above (\cite{downey1995fixed}).
	
	The details of the reduction follow.
	Given an instance of the $d$-sum problem with $s$ integers $x_1,x_2,\cdots,x_s$, we construct an instance of the KU problem as follows. Let $n'=s$ and $r=2d-1$. Construct $n'$ items in the KU problem with $p_i=\Pr(P_i=1)=2^{-\omega x_i}$ and $q_i=M-\omega x_i$ for $1\leq i \leq n$, where $\omega=\lceil \log_2 \alpha\rceil+1$ and $M=s\omega\sum_{i=1}^sx_i$. Let $Q=dM-\omega t$.
	We make two claims.
	
	\begin{claim}
		\label{claim-new-1}
		If the $d$-sum instance admits a solution, then there exists a  solution to the KU problem with an objective value at least $2^{-\omega t}$.
	\end{claim}
	
	\begin{proof}
		%[{\bf Proof of Claim \ref{claim-new-1}}]
		Suppose the $d$-sum problem admits a  solution $E$. Let $I=\{i|x_i\in E\}$ be the index set of items in the solution. We observe that 
		\begin{eqnarray*}
			\sum_{i\in I} q_i&=&dM-\omega\sum_{i\in I}x_i=dM-\omega t=Q, ~\text{and}\\
			\Pr\left(\sum_{i\in I}P_i\ge d\right)&=&\Pr\left(P_i=1, \forall i\in I\right)=\prod_{i\in I}p_i=2^{-\omega t}.
		\end{eqnarray*}
		Hence, there exists a  solution with an objective value at least $2^{-\omega t}$.
		Thus, {\bf Claim \ref{claim-new-1}} holds.
	\end{proof}
	
	\begin{claim}
		\label{claim-new-2}
		If the $d$-sum instance does {\em not} admit a solution, then any  solution to the KU problem has an objective value at most $2^{-\omega (t+1)}<1/\alpha\cdot 2^{-\omega t}$.
	\end{claim}
	
	\begin{proof}
		%[{\bf Proof of Claim \ref{claim-new-2}}]
		Suppose the $d$-sum problem does not admit a solution. 
		Note that for any solution $I$ to the KU problem, we have $|I|\le d$; otherwise, $|I|\ge d+1$ leads to $$\sum_{i\in I} q_i\ge (d+1)M-\omega\sum_{i\in I}x_i>dM>Q,$$
		which contradicts that $I$ is a  solution. We split $|I|\le d$ into two scenarios: $|I|\le d-1$ or $|I|= d$.
		\begin{itemize}
			\item In the case $|I|\le d-1$, {\bf Claim \ref{claim-new-2}} holds because
			\begin{eqnarray*}
				&&\Pr\left(\sum_{i\in I}P_i\ge r+1-|I|\right)\le \Pr\left(\sum_{i\in I}P_i\ge d+1\right)\\
				&=&0<2^{-\omega(t+1)}.
			\end{eqnarray*}
			
			\item In the case $|I|=d$, the fact $\sum_{i\in I}q_i\le Q$ and $q_i=M-\omega x_i$ and $Q=dM-\omega t$ implies $\sum_{i\in I} x_i\ge t$. Since the $d$-sum problem does not admit a  solution, either $\sum_{i \in I} x_i\ge t+1$ or $\sum_{i:x_i\in I} x_i\le t-1$. Given that $\sum_{i\in I} x_i\ge t$, we have $\sum_{i\in I} x_i\ge t+1$. Then, {\bf Claim \ref{claim-new-2}} holds because
			$$\Pr\left(\sum_{i\in I}P_i\ge d\right)=\prod_{i\in I}p_i=2^{-\omega\sum_{i\in I}x_i}\le 2^{-\omega (t+1)}.$$
		\end{itemize}
		%Thus, {\bf Claim \ref{claim-new-2}} holds.
		%Therefore, Theorem~\ref{thm:hardness} holds. 	
	\end{proof}
	
	Under {\bf Claims \ref{claim-new-1}}-{\bf \ref{claim-new-2}}, we observe that an $\alpha$-approximation algorithm for the KU problem can be used to solve the $d$-sum problem as follows:
	\begin{itemize}
		\item In the case the $\alpha$-approximation algorithm for the KU porblem returns a feasible solution with an objective value that is $\leq 2^{-\omega (t+1)}$, the optimal objective value is at most $\alpha\cdot 2^{-\omega (t+1)}< 2^{-\omega t}$. {\bf Claim \ref{claim-new-1}} implies that the $d$-sum instance does not admit a feasible solution.
		\item In the case the $\alpha$-approximation algorithm for the KU problem returns a feasible solution with an objective value that is $> 2^{-\omega (t+1)}$, {\bf Claim \ref{claim-new-2}} implies that the $d$-sum instance must admit a feasible solution. 
	\end{itemize}
	Hence, any $\alpha$-approximation algorithm for solving the KU problem can be used to solve the $d$-sum problem. 
	%Thus, Theorem~\ref{thm:hardness} holds.
	This completes the proof of {\bf Theorem~\ref{thm:hardness}}.
\end{proof}

%As a corollary of {\bf Theorem~\ref{thm:hardness}}, we obtain:

\begin{coro}%[{\color{blue}another form of hardness of the KU problem}]
	\label{coro:apx}
	Assuming $P\neq NP$, there does not exist an $O(1)$-approximation algorithm for the KU problem that runs in polynomial time if $r$ is part of the input.
\end{coro}

Now we are ready to prove {\bf Theorem \ref{coro:pec-c-u}}.

\begin{proof}[{\bf Proof of Theorem \ref{coro:pec-c-u}}]
	{\bf Lemma \ref{lemma:BU-KU-equivalence}} shows that the KU problem is equivalent to the BVU problem with two candidates. The hardness of the KU problem is established by {\bf Theorem \ref{thm:hardness}} and {\bf Corollary \ref{coro:apx}}. Hence {\bf Theorem \ref{coro:pec-c-u}} holds.
\end{proof}

\section{An Approximation Algorithm in FPT time}
\label{sec:FPT-Approximation-algorithms}

Having showed that the BVU problem is hard, now we present an approximation algorithm for solving it.
The algorithm
runs in FPT time for any fixed constant $m$ and any small constant $\epsilon$.
In terms of approximation ratio, our algorithm returns a value that is
$\geq \OPT-\epsilon$, where $\OPT\in [0,1]$ is the optimal objective value.
Note that the hardness result given by {\bf Theorem~\ref{coro:pec-c-u}} suggests that an additive approximation algorithm is perhaps the best algorithm we can hope for.

\subsection{Algorithmic Result}

\begin{theorem}[Main algorithmic result]\label{thm:BU-alg}
	For an arbitrary small constant $\epsilon>0$,
	there exists an algorithm for the BVU problem, which runs
	in $r^{O(mr/\epsilon)}+n^{O(m^5/\epsilon^{5})}$ time and returns a solution with {an objective value} no smaller than $\OPT-\epsilon$, where
	$\OPT\in [0,1]$ is the optimal objective value.
\end{theorem}

In order to prove {\bf Theorem \ref{thm:BU-alg}}, we need to design an approximation algorithm for the BVU problem. For this purpose, we define a new variant of the Knapsack problem.
%, called Multi-block Knapsack under Uncertainty (MKU).

\smallskip

\noindent\textbf{The MKU Problem.}
The MKU problem deals with items that have deterministic sizes but random profits and involves a stochastic objective function,
and the goal is to maximize a certain ``overflow" probability under the knapsack's volume and cardinality constraints.
More specifically, the MKU problem is defined as follows:

\begin{center}
	\fbox{\begin{minipage}{0.47\textwidth}
			\textbf{Multi-block Knapsack with Uncertainty (MKU)}
			
			Input: A knapsack of capacity $Q$; a set of items $\mathcal{V}=\{v_1,v_2,\cdots, v_n\}$, with each item associated with a size $q_i$ and a profit $P_i$, which is an independent random variable such that $\Pr(P_i=1)=p_i$ and $\Pr(P_i=0)=1-p_i$; a partition of the $n$ items into a constant $m\ge 2$ subsets $V_1,V_2,\cdots,V_{m}$, and a quota $\Delta_j$ for each $V_j$ such that $\Delta_j\le r+1$ for some positive integer $r$; a positive integer $k$ such that $k\le r+1$; a positive index $1\le j_0\le m-1$.
			
			Output: Find a set of indices $I^*\subseteq \{1,2,\cdots,n\}$ such that
			\begin{itemize}
				\item $\sum_{i\in I^*} q_i\le Q$,
				\item $|V(I^*)\cap V_j|\ge \Delta_j$ for all $1\le j\le m-1$ and $j\neq j_0$,
				\item $|V(I^*)\cap V_{j_0}|= \Delta_{j_0}$,
				\item $\Pr\left(\sum_{i\in I^*} P_i\ge k\right)$ is maximized,
			\end{itemize}
			where $V(I^*)=\{v_i|i\in I^*\}$.
	\end{minipage}}
\end{center}

Note that in the preceding definition, we intentionally make the parameters of the MKU problem correspond to the parameters of the BVU problem exactly, because we intend to reduce the number of notations used in this paper (for better readability).
That is, parameters $n$, $m$, $V_j$, $Q$, $p_i$, $q_i$, $r$ and $I^*$ in the BVU problem correspond to the same parameters
%$n$, $m$, $V_j$, $Q$, $p_i$, $q_i$, and $r$
in the MKU problem.
We will use the problem context to distinguish the meanings of these parameters.
Because of this, we say an instance of the MKU problem corresponds to an instance of the BVU problem when they have the same set of parameter values.
%, $m$, $V_j$, $Q$, $p_i$, $q_i$, and $r$ are .
%and $k\le r$.

%\begin{proof}[{\bf Proof of Theorem~\ref{thm:BU-alg}}]
Now we show that the BVU problem can be solved efficiently by utilizing an algorithm for the MKU problem.
% (Section~\ref{subsec:MKP}).
%\end{proof}

\begin{theorem}%[reduction from the BVU problem to the MKU problem]
	\label{thm:reduction-mku}
	Let $\epsilon>0$ be an arbitrary small constant.
	Denote by $OPT_{BVU}$ and $OPT_{MKU}$ the optimal objective value of the BVU problem and the MKU problem, respectively.
	A feasible solution to the BVU problem with
	an objective value at least $OPT_{BVU}-\epsilon$ can be found in $O(rm\Lambda)$ time, where $\Lambda$ is the time for finding a {feasible} solution to the {corresponding} MKU problem with the objective value at least $OPT_{MKU}-\epsilon$.
\end{theorem}

\begin{proof}
	Let $I$ be an arbitrary solution to the BVU problem, $V(I)=\{v_i|i\in I\}$ and $V'_j=V_j\cap V(I)$. 
	%{\color{red}for $1\leq j \leq m$, while noting that $V'_m=\emptyset$ (because the attacker does not need to bribe the voters that prefer the designated candidate already)}. 
	For any $v_i\in {V'_j}$, we define $X_i^{j}$ to be a binary random variable indicating whether $v_i$ votes or not if it is bribed, i.e., $\Pr(X_i^j=1)=p_i$ and $\Pr(X_i^j=0)=1-p_i$. 
	
	%As $c_{m}$ is the designated candidate, there is no need to bribe voters who have already voted for $c_{m}$. 
	For $j=1,\ldots,m-1$, if $v_i\in V_j$ is bribed { (i.e., $v_i\in V'_j$)}, then there are two scenarios:
	\begin{itemize}
		\item The bribery succeeds, meaning that the number of votes received by $c_j$ decreases by 1 and the number of votes received by $c_{m}$ increases by 1.
		\item The bribery fails, meaning that    the number of votes received by $c_j$ decreases by 1 but the number of votes received by $c_{m}$ remains unchanged. 
	\end{itemize}
	Let $Y_j$ be the total number of votes received by $c_j$ after bribing. Then, we have 
	\begin{subequations}
		\begin{eqnarray*}
			Y_j  &=&|V_j|-|V_j'|, ~~~~1 \leq j \leq m-1;\\
			Y_{m}&=&|V_{m}|+\sum_{1 \leq j \leq m-1} ~ \sum_{v_i\in V_j'} X_i^j.
		\end{eqnarray*}
	\end{subequations}
	Note that $Y_j$ for $1 \leq j\leq m-1$ is a deterministic value, rather than a random variable, because no matter the bribery of $v_i\in V_j'$ succeeds or not, $c_j$ always loses the vote of $v_i$. 
	%For $j\neq m$, let $EV_j$ be the event that $j$ receives less votes than $m$ after bribery, 
	The probability that the designated candidate $c_m$ becomes the winner is: 
	
	\begin{eqnarray*}
		&&\Pr\left(Y_j<Y_{m},  1\leq j\leq m-1\right)\\
		&=&\Pr\left(\sum_{1 \leq \ell \leq m-1}~\sum_{i:v_i\in V_\ell'} X_i^\ell>|V_j|-|V_j'|-|V_{m}|,  1 \leq j \leq m-1\right)\\
		&=& \Pr\left(\sum_{1 \leq \ell \leq m-1}~\sum_{i:v_i\in V_\ell'} X_i^\ell>\max_{1 \leq j \leq m-1}\left(|V_j|-|V_j'|-|V_{m}|\right)\right).
	\end{eqnarray*}
	
	We observe that this probability is only dependent on the value of $\max\limits_{1 \leq j \leq m-1}\left(|V_j|-|V_j'|-|V_{m}|\right)$ and that $X_i^j\ge 0$ for $1 \leq j \leq m-1$. If $\max\limits_{1 \leq j \leq m-1}\left(|V_j|-|V_j'|-|V_{m}|\right)\le -1$, then we have
	\begin{eqnarray*}
		&&\Pr\left(\sum_{1 \leq j \leq m-1}~\sum_{i:v_i\in V_j'} X_i^j>\max_{1 \leq j \leq m-1}\left(|V_j|-|V_j'|-|V_{m}|\right)\right)\\
		&=&\Pr\left(\sum_{1 \leq j \leq m-1}~\sum_{i:v_i\in V_j'} X_i^j>-1\right)=1.
	\end{eqnarray*}
	%Based on the preceding discussion, 
	For any solution $I$ to the BVU problem, we define:
	\begin{eqnarray*}
		\xi(I)=
		\left\{
		\begin{array}{ll}
			\max\limits_{1 \leq j \leq m-1}\left(|V_j|-|V_j'|-|V_{m}|\right), & \\
			~~~~~~~~~~~~~~~~\text{if}~ \max\limits_{1 \leq j \leq m-1}\left(|V_j|-|V_j'|-|V_{m}|\right)\ge 0; \\
			-1, ~~~~~~~~~~~ \text{otherwise}
		\end{array}
		\right. 
	\end{eqnarray*}
	Then our previous arguments show that the BVU problem can be reformulated as follows: 
	\begin{center}
		\fbox{\begin{minipage}{0.47\textwidth}
				Find $I^*\subseteq \{1,2,\cdots,n\}$ such that $\sum_{i\in I} q_i\le Q$ and $\Pr\left(\sum\limits_{1 \leq j \leq m-1}~\sum\limits_{i:v_i\in \in V_j\cap V(I^*)} X_i^j\ge \xi(I^*)+1\right)$ is maximized.
		\end{minipage}}
	\end{center}
	Recall that candidate $c_1$ is the winner in the absence of bribery and $|V_{m}|=|V_1|-r$. Then, for any $1 \leq j \leq m-1$, we have $|V_j|\le |V_{m}|+r$ and therefore $\xi(I)\le r$ for any feasible solution $I$, leading to $\xi(I)\in\{-1,0,1,\cdots,r\}$, and in particular $\xi(I^*)\in\{-1,0,1,\cdots,r\}$.
	Hence, we can guess the value of $\xi(I^*)$. When we guess the value of $\xi(I^*)$ correctly, say, $\xi(I^*)=\alpha$, then the BVU problem is equivalent to finding some $I^*$ such that $\sum_{i\in I^*} q_i\le Q$, $\xi(I^*)= \alpha$ and 
	$$\Pr\left(\sum_{1 \leq j \leq m-1}~\sum_{i:v_i\in V_j'} X_i^j\ge \alpha+1\right)$$ 
	is maximized. By definition, $\xi(I^*)= \alpha$ holds if and only if the following two conditions are simultaneously satisfied:
	\begin{itemize}
		\item {\bf Condition 1:} There exists some $j_0\in\{1,\ldots,m-1\}$ such that $|V_{j_0}|-|V_{j_0}'|-|V_{m}|= \alpha$.
		\item {\bf Condition 2:} For any $j\in\{1,\ldots,m-1\}$, we have $|V_j|-|V_j'|-|V_{m}|\le \alpha$.
	\end{itemize}
	Let us define
	\begin{eqnarray*}
		\Delta_j=
		\left\{
		\begin{array}{ll}
			|V_j|-|V_{m}|-\alpha & \text{if}~~ |V_j|-|V_{m}|-\alpha\ge 0;\\
			0 & \text{otherwise}.
		\end{array}
		\right.
	\end{eqnarray*}
	Then, the preceding {\bf Conditions 1} and {\bf 2} are equivalent to $|V_{j_0}\cap V(I)|=|V_{j_0}'|=\Delta_{j_0}$ and $|V_j\cap V(I)|\ge \Delta_j$, {respectively}. 
	
	Hence, when we guess $\xi(I^*)$ and $j_0$ correctly, the BVU problem is exactly the same as the MKU problem, whereas a (near-)optimal solution to the MKU problem implies a (near-)optimal solution to the BVU problem.
	Since guessing $\xi(I^*)$ and $j_0$ takes $O(rm)$ enumerations, we can solve the BVU problem by having oracle access to an algorithm that solves the MKU problem. 
	Theorem~\ref{thm:reduction-mku} is proved.
\end{proof}

\begin{proof}[{\bf Proof of Theorem~\ref{thm:BU-alg}}]
	{\bf Theorem \ref{thm:reduction-mku}}, which is stated and proven below, shows that a (approximate) solution of the BVU problem can be found in polynomial oracle-time, by utilizing a (approximation) algorithm for the MKU problem as an oracle.
	%solved in polynomial-time by using the algorithm for the MKU problem as an oracle.
	% (Section \ref{subsec:reduction}).
	The remaining task is to design an approximate algorithm for solving the MKU problem, which is quite involved and suggests us to 
	%In the second step, we algorithmically solve with the MKU problem (Section~\ref{subsec:MKP}). 
	%The strategy for solving the MKU problem is to differentiate between big items from small items, which will be handled separately. 
	use the ``divide and conquer'' strategy by considering two cases separately ({\bf Theorem \ref{thm:fpt-alg}}).
	% (Section~\ref{subsec:MKP}).
\end{proof}

Now we show that there is an approximate algorithm for solving the MKU problem.
% ({\bf Theorem \ref{thm:fpt-alg}}).
\begin{theorem}[algorithm for solving the MKU problem]
	\label{thm:fpt-alg}
	For any arbitrary small constant $\epsilon>0$, %{\color{blue}and a constant $m$},
	there exists an algorithm for the MKU problem that runs
	in $r^{O(mr/\epsilon)}+n^{O(m^5/\epsilon^{5})}$ time and returns a solution with an objective value that is no smaller than $\OPT-\epsilon$, {where
		$\OPT\in [0,1]$ is the optimal objective value in the MKU problem.}
\end{theorem}

\begin{proof}[{\bf Proof of Theorem~\ref{thm:BU-alg}}]
	%{\bf Theorem \ref{thm:reduction-mku}} showed that the BVU problem can be solved efficiently by utilizing an algorithm for the MKU problem. {\bf Theorem \ref{thm:fpt-alg}} showed that there is an approximate algorithm for solving the MKU problem. Therefore, {\bf Theorem~\ref{thm:BU-alg}} holds. 
	By putting {\bf Theorem \ref{thm:reduction-mku}} and {\bf Theorem \ref{thm:fpt-alg}} together, we obtain {\bf Theorem~\ref{thm:BU-alg}}.
\end{proof}

%\subsection{Designing Algorithm for the MKU Problem}
\subsection{The Proof of Theorem \ref{thm:fpt-alg}}
~\label{subsec:MKP}

%The proof of {\bf Theorem \ref{thm:fpt-alg}} is quite complicated. 
The main difficulty originates from the maximization of a probability involving the sum of random variables, which does not have a simple explicit expression. A natural idea is to approximate the summation of random variables with a Gaussian variable via Berry-Essen's Theorem. However, such an approximation is not always achievable because the condition in Berry-Essen's Theorem does not necessarily hold. Furthermore, even if Berry-Essen's Theorem is applicable, bounding the tail probability of a Gaussian variable together with a set of other constraints required in MKP is still challenging. Figure \ref{fig:illustration} highlights the proof strategy for overcoming these difficulties.

\begin{figure}[!htbp]
	\centering
	\fbox{
		\begin{minipage}{.45\textwidth}
			\centering
			\includegraphics[width=\textwidth]{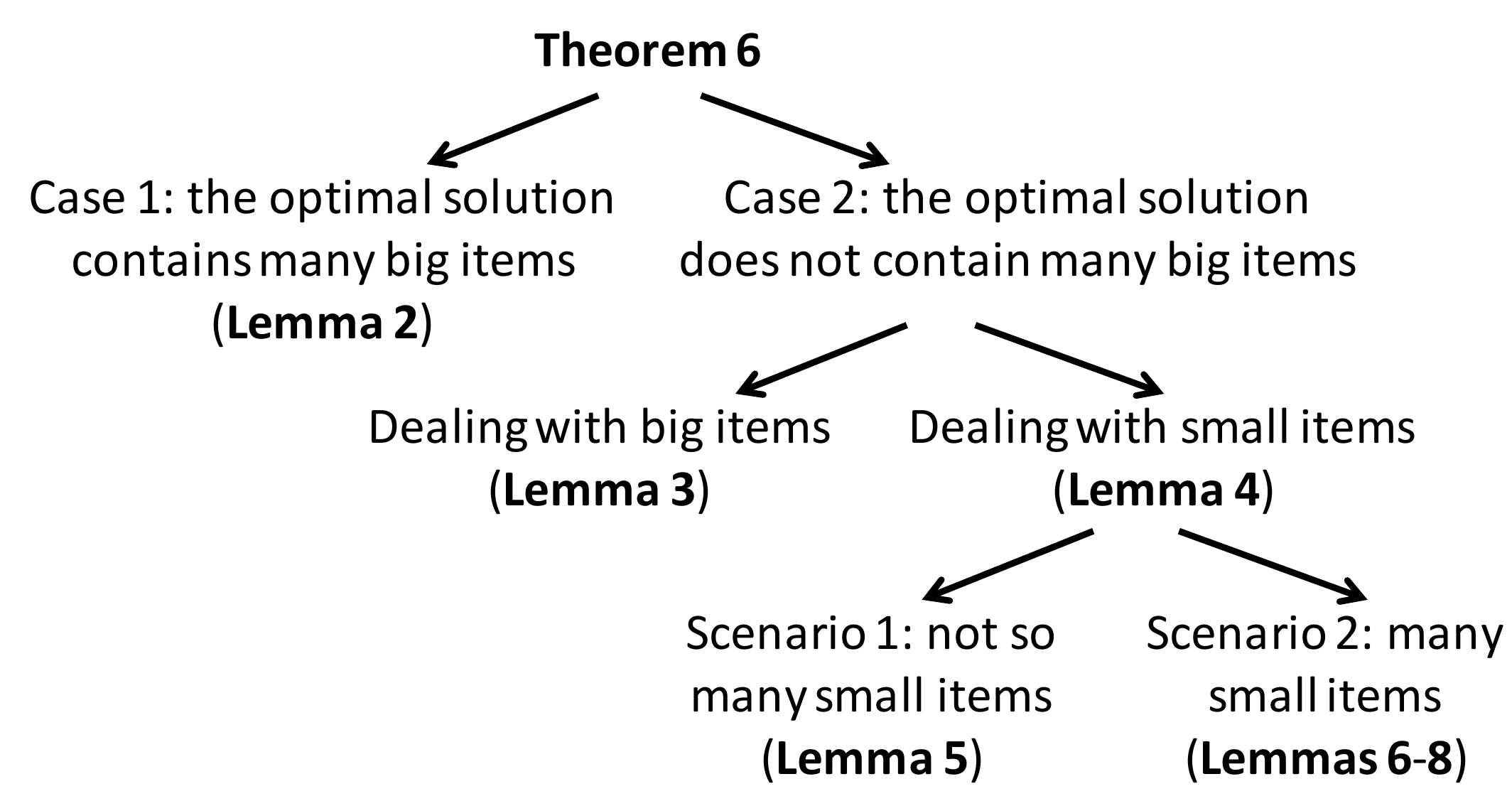}
			\caption{Strategy for proving {\bf Theorem \ref{thm:fpt-alg}}.\label{fig:illustration}}
		\end{minipage}
	}
\end{figure}

Specifically, We partition the set of items into {\em big} and {\em small} ones based on their probability. Then,
we differentiate the case that the optimal solution contains many big items ({\bf Case 1}),
which is easily coped with by using Markov's inequality ({\bf Lemma \ref{obs:1}}), from the case that the optimal solution does not contain many big items ({\bf Case 2}),
whose treatment is much more complicated and proceeds as follows.
%therefore further divided into the following three steps:
\begin{itemize}
	\item First, we apply {\bf Corollary~\ref{coro:folklore-sequence}} to decompose the MKU problem in {\bf Case 2} into a series of sub-problems, each of which is a stochastic knapsack problem with one cardinality constraint.
	\item Second, for big items ({\bf Lemma~\ref{lemma:fpt-big}}), we round their probability to $O(k/\epsilon)\le O(r/\epsilon)$ distinct probabilities. This allows us to guess the number of big items corresponding to the rounded probabilities in the optimal solution, leading to the selection of the optimal subset of big items.
	\item Third, for small items ({\bf Lemma~\ref{lemma:small}}), there are two scenarios:
	{%\color{red}\footnote{the text description is not consistent with the description in the picture; please fix}
		\begin{itemize}
			\item In the scenario where the optimal solution does not contain a large volume of small items, we present a dynamic programming algorithm ({\bf Lemma~\ref{lemma:dp}}). %while leveraging Markov's inequality.
			\item In the scenario where the optimal solution contains a large volume of small items, Berry-Essen's Theorem is applicable and we can use it to transform the problem of maximizing a specific probability to the problem of approximating the summation of moments of random variables in the optimal solution. Since the moments of a random variable are deterministic, we can leverage the technique for solving the classic knapsack problem ({\bf Lemmas~\ref{lemma:transform}-\ref{lemma:small-2}}).
		\end{itemize}
	}
\end{itemize}

%Now we proceed to present all the ingredients needed for the proof of {\bf Theorem \ref{thm:fpt-alg}}.

\begin{definition}[big vs. small items]
	Under the assumption that $\epsilon>0$ is a small constant such that $1/\epsilon\ge 4$ is an integer, we say an item in the MKU problem is {\em big} if $p_i>1-\epsilon^2$ and is {\em small} otherwise.
\end{definition}

%Now we prove all the lemmas needed for the above proof.

\begin{lemma}[the case the optimal solution containing many big items]\label{obs:1}
	If $|T\cap B_{j^*}|\ge 2k$, then there is a polynomial-time algorithm that returns a solution $I$ to the MKU problem such that
	\begin{itemize}
		\item $\sum_{i\in I} q_i\le Q$,
		\item $|V(I)\cap V_j|\ge \Delta_j$ for $j\neq j_0$,
		\item $|V(I)\cap V_{j_0}|= \Delta_{j_0}$, and
		\item $\Pr\left(\sum_{i\in I} P_i\ge k\right)\ge 1-\epsilon$.
	\end{itemize}
	%    $$\Pr(P_{I}\ge k)\ge 1-\epsilon ~\text{ and }~ Q_{I}\le Q_{T\cap B}\le Q.$$
\end{lemma}
\begin{proof}
	We first select the $2k$ big items with the smallest sizes within $B_{j^*}$. Among the remaining items in each $V_j$, we further select the items of the smallest size to make sure that we have selected at least $\Delta_j$ items in each $V_j$ and exact $\Delta_{j_0}$ items in $V_{j_0}$ -- this can be achieved by a simple greedy strategy, that is, we check if there is any $j$ such that the cardinality constraint is not satisfied yet, and pick items of smallest size in $V_j$ to ensure the cardinality constraint. Let $I$ be the set of items that are selected as such, then obviously we have $|I\cap V_j|\ge \Delta_j$ and $|I\cap V_{j_0}|= \Delta_{j_0}$. Since we always select the smallest items, we have $\sum_{i\in I} q_i\le Q$. That is, the first three conditions required by Lemma \ref{obs:1} are satisfied.
	
	In what follows we show the last condition, namely $\Pr\left(\sum_{i\in I} P_i\ge k\right)\ge 1-\epsilon$.
	%{\color{red}Now we estimate $\Pr(P_{I}\ge k)$ as follows.}
	Let $X_i=1-P_i$ and $\mu=\EE\left(\sum_{i\in I} X_i\right)\le 2k\epsilon^2$. By applying Markov's inequality Eq. \eqref{eq:markov-inequality}, we have
	\begin{eqnarray*}
		\Pr(P_I<k)&=& \Pr\left(\sum_{i\in I} X_i\ge k+1\right)\\
		&\le & \Pr\left(\sum_{i\in I} X_i\ge \mu \cdot \frac{1}{2\epsilon^2}\right)\\
		&\le & 2\epsilon^2.
	\end{eqnarray*}    
	Hence, $\Pr(P_I\ge k)\ge 1-\epsilon$. The completes the proof of the Lemma \ref{obs:1}.
\end{proof}

\begin{lemma}[dealing with big items in the case the optimal solution not containing many big items]\label{lemma:fpt-big}
	If $|T\cap B_{\ell}|\leq 2k-1$, then there is an algorithm that runs in $k^{O(mk/\epsilon)}\le r^{O(mr/\epsilon)}$ time and returns a set $I\cap B_{\ell}$ of big items such that
	\begin{itemize}
		\item $|I\cap B_{\ell}|=|T\cap B_{\ell}|$,
		\item $Q_{I\cap B_{\ell}}\le Q_{T\cap B_{\ell}}$, and
		\item $\Pr(P_{I\cap B_{\ell}}\ge h)\ge (1-2\epsilon/m)\Pr(P_{T\cap B_{\ell}}\ge h)$ for any $h\ge 0$.
	\end{itemize}
\end{lemma}
\begin{proof}
	{We round the probabilities associated to the big items as follows.}
	%Let $n_B=|B|$. For simplicity we re-index all items such that item $1$ to item $n_B$ are the big items.
	Let $\delta=\epsilon/(mk)$ and $\beta=O(1/\delta)=O(mk/\epsilon)$ be the largest integer such that $(1-\epsilon^2)(1+\delta)^{\beta}< 1$. Let
	$$\Gamma_1=\{1-\epsilon^2,(1-\epsilon^2)(1+\delta),\ldots,(1-\epsilon^2)(1+\delta)^{\beta}\}$$
	be the set of rounded probabilities. For each big item, we round its probability $p_i$ down to the nearest value in $\Gamma_1$ and denote it by $\tilde{p}_i$. Note that $\tilde{p}_i\le p_i<\tilde{p}_i(1+\delta)$. Let $B_{\ell}^s$ be the subset of big items such that their associated probabilities are rounded to $(1-\epsilon^2)(1+\delta)^{s}$.
	
	For each $\ell\le O(k/\epsilon)$, we guess the value of $|T\cap B_{\ell}^s|\le O(k)$. There are at most $k^{O(mk/\epsilon)}$ different possibilities on these values. Once we guess $|T\cap B_{\ell}^s|$ correctly for each $\ell$, we select the $|T\cap B_{\ell}^s|$ items that have the smallest size in $B_{\ell}^s$ and let $\bar{B}_\ell^{s}$ denote the set of these items. 
	Set $I\cap B_{\ell}=\cup_{s=0}^{\beta}\bar{B}_\ell^{s}$.
	
	{Now we prove that $I\cap B_{\ell}$ defined above satisfies Lemma~\ref{lemma:fpt-big}, whereas the lemma is proved}. %Therefore, we can define $I\cap B=\cup_{j=0}^{\beta}\bar{B}_j$ and claim that $I\cap B$ satisfies {the condition required by} Lemma~\ref{lemma:fpt-big}. 
	For this purpose, we first observe that $|I\cap B_{\ell}^s|=|T\cap B_{\ell}^s|$ and $I\cap B_{\ell}^s$ always consists of the items with the smallest size in $B_{\ell}^s$, therefore we have $Q_{I\cap B_{\ell}}\le Q_{T\cap B_{\ell}}$. Then, we compare $\Pr(P_{I\cap B_{\ell}}=h)$ and $\Pr(P_{T\cap B_{\ell}}=h)$ for every $h\ge 0$. Let $\phi$ be an arbitrary one-to-one mapping that maps each item in $T\cap B_{\ell}$ to a distinct item in $I\cap B_{\ell}$ for every $j$. We have
	\begin{eqnarray*}
		\Pr(P_{T\cap B_{\ell}}=h)&=&\sum_{I\subseteq T\cap B_{\ell}, |I|=h}\prod_{i\in I}p_i\prod_{i\not\in I}(1-p_i),\\
		\Pr(P_{I\cap B_{\ell}}=h)&=&\sum_{I\subseteq T\cap B_{\ell}, |I|=h}\prod_{i\in I}p_{\phi(i)}\prod_{i\not\in I}(1-p_{\phi(i)}).
	\end{eqnarray*}
	In order to show  $\Pr(P_{I\cap B_{\ell}}=h)\ge (1-2\epsilon/m)\Pr(P_{T\cap B_{\ell}}=h)$, it suffices to show that
	$$\prod_{i\in I}p_{\phi(i)}\prod_{i\not\in I}(1-p_{\phi(i)})\ge (1-2\epsilon/m)\prod_{i\in I}p_i\prod_{i\not\in I}(1-p_i)$$
	for every $I\subseteq T\cap B_{\ell}$ {with} $|I|=h$. According to the way we round probabilities, we have $p_{\phi(i)} \le p_i<p_{\phi(i)}(1+\delta)$, hence $1-p_{\phi(i)}\ge 1-p_i$ and
	$$\prod_{i\in I}p_{\phi(i)}\ge (1-\delta)^{h}\prod_{i\in I}p_i\ge (1-h\delta)\prod_{i\in I}p_i.$$
	For $h\le 2k$, we have $h\delta\le 2\epsilon/m$ and
	$$\Pr(P_{I\cap B_{\ell}}=h)\ge (1-2\epsilon/m)\Pr(P_{T\cap B_{\ell}}=h).$$
	For $h> 2k$, $\Pr(P_{T\cap B_{\ell}}=h)=0$ and the above inequality is trivially true.
	
	This completes the proof of {\bf Lemma~\ref{lemma:fpt-big}}.
\end{proof}

\begin{lemma}[dealing with small items in the case the optimal solution not containing many big items]\label{lemma:small}
	There exists an algorithm that runs in $n^{O(m^5/\epsilon^5)}$ time and returns a feasible solution $I\cap S_\ell$ such that
	\begin{itemize}
		\item $|I\cap S_{\ell}|=|T \cap S_\ell|$
		\item $Q_{I\cap S_\ell}\le Q_{T\cap S_\ell}$
		\item $\Pr(P_{I\cap S_\ell}\ge h)\ge \Pr(P_{T\cap S_\ell}\ge h)-\Theta(\epsilon/m)$
	\end{itemize}
\end{lemma}

The proof of this lemma needs a sequence of results.

\begin{lemma}\label{lemma:dp}
	For any $\zeta$, there exists an algorithm that runs in $(mn/\epsilon)^{O(\zeta)}$ time and returns a solution $I\cap S_{\ell}$ such that $|I\cap S_{\ell}|= |T\cap S_\ell|$, $Q_{I\cap S_{\ell}}\le Q_{T\cap S_{\ell}}$ and $\Pr(P_{I\cap S_{\ell}}=h)\le \Pr(P_{T\cap S_{\ell}}=h)+2\epsilon/(mn)$ for every $0\le h\le \zeta-1$.
	%i.e., $I$ is a feasible solution with an objective value at least $OPT-\epsilon$.
\end{lemma}
\begin{proof}[{\bf Proof of Lemma \ref{lemma:dp}}]
	%{\color{blue}The proof strategy is the following.}
	We design an algorithm based on dynamic programming. Let $\eta=\epsilon/(mn^2)$. Although we do not know the value of $\Pr(P_{T\cap S_{\ell}}=h)$, we know that this value lies in $[0,1]$. Therefore, we can guess, via $(mn/\epsilon)^{O(\zeta)}$ enumerations, the $\zeta$ values $t_0,t_1,\cdots,t_{\zeta-1}$ such that $t_h-\epsilon/(mn)\le \Pr(P_{T\cap S_{\ell}}=h)< t_h$. In the following we provide an algorithm that returns $I\cap S_{\ell}$ such that $\Pr(P_{I\cap S_{\ell}}=h)\le t_h+\epsilon/(mn)$, and {Lemma~\ref{lemma:dp} follows.}
	%Let $\beta=O(\frac{n^2\log n}{\epsilon})$ be the smallest integer such that $\eta(1+\eta)^\beta\ge 1$.
	
	{Let us define $\Gamma_2=\{0,\eta,2\eta,\ldots,\eta\cdot 1/\eta\}$ as the set of rounded probabilities.}
	Let us call a $(\zeta+2)$-vector $(j,\varpi,u_0,u_1,u_2,\ldots,u_{\zeta-1})$ a {\em state}, where $j\in\{0,1,\ldots,n'\}$, $\varpi\in\{0,1,\cdots,n'\}$ and $u_j\in \Gamma_2$. Each state is associated with a positive value $f(j,\varpi,u_0,u_1,\ldots,u_{\zeta-1})$, which can be calculated recursively as shown in the next paragraph. Intuitively, a state means that a subset $U\subseteq \{1,2,\ldots,j\}$ of items can be selected such that $\Pr(P_U=j)$ is approximately $u_j$, $|U|=\varpi$ and $f(j,\varpi,u_0,u_1,\ldots,u_{\zeta-1})$ is the minimal total size of items among all possible subsets $U$. In particular, if such a subset $U$ does not exist, then $f(j,\varpi,u_0,u_1,\ldots,u_{\zeta-1})=\infty$.
	
	\smallskip
	
	Now we define the {calculation} of $f(j,\varpi,u_0,u_1,\ldots,u_{\zeta-1})$. For this purpose, we first define the summation of state $(j,\varpi,u_0,u_1,\ldots,u_{\zeta-1})$ and random variable $P_{j+1}$ as follows:
	\begin{eqnarray}\label{eq:dp}
	(j,\varpi,u_0,u_1,\ldots,u_{\zeta-1})+P_{j+1}\nonumber\\=(j+1,\varpi+1,\tilde{u}_0',\tilde{u}_1',\ldots,\tilde{u}_{\zeta-1}'),
	\end{eqnarray}
	where %$\varpi'=\min\{\varpi+1,\Delta_{\ell}-|I\cap B_\ell|\}$, 
	$\tilde{u}_j'$ is the nearest value in $\Gamma_2$ when rounding {\em up} $u_j'$ with
	$u_0'=u_0(1-p_{j+1})$ and $u_{j}'=u_j(1-p_{h+1})+u_{j-1}p_h$ for $1\le j\le \zeta-1$.

	Initially, for $j=0$, we define
	\begin{equation}
	\label{eq:recursive-definition-h-1}
	f(0,0,0,\ldots,0)=0
	\end{equation}
	and  { for $(u_0,u_1,\ldots,u_{\zeta-1})\neq (0,0,\ldots,0)$, we define}
	\begin{equation}
	\label{eq:recursive-definition-h-2}
	f(0,\varpi,u_0,u_1,\ldots,u_{\zeta-1})=\infty.
	\end{equation}

	For $h\ge 0$, we define
	\begin{eqnarray*}
		g(j+1,\varpi,u_0,u_1,\ldots,u_{\zeta-1})=\\
		q_{j+1}+\min\{f(j,\varpi',u_0',u_1',\ldots,u_{\zeta-1}'):\\
		(j,\varpi',u_0',u_1',\ldots,u_{\zeta-1}')+P_{j+1}=(j+1,\varpi,u_0,u_1,\ldots,u_{\zeta-1}) \},
	\end{eqnarray*}
	and define
	\begin{eqnarray}\label{eq:recursive-1}
	f(j+1,\varpi,u_0,u_1,\ldots,u_{\zeta-1})=\nonumber\\
	\min\{f(j,\varpi,u_0,u_1,\ldots,u_{\zeta-1}),g(j+1,\varpi,u_0,u_1,\ldots,u_{\zeta-1})\}.
	\end{eqnarray}
	
	%\begin{eqnarray}\label{eq:recursive-1}
	%    &&\hspace{-5mm}f(h+1,u_0,u_1,\cdots,u_{\zeta-1})=\min\{f(h,u_0,u_1,\cdots,u_{\zeta-1}),\nonumber\\&&\hspace{-5mm}f(h,u_0',u_1',\cdots,u_{\zeta-1}')+q_{h+1} \nonumber\\ &&\hspace{-5mm} {\color{red}\text{ if } (h+1,u_0,u_1,\cdots,u_{\zeta-1})=(h,u_0',u_1',\cdots,u_{\zeta-1}')+P_{h+1}}\}
	%    \end{eqnarray}
	Observe that we can use Eqs. \eqref{eq:recursive-definition-h-1}-\eqref{eq:recursive-1} to recursively calculate the value associated to any state. Since the total number of states is bounded {from above} by $|\Gamma_2|^{O(\zeta)}=(mn/\epsilon)^{O(\zeta)}$, the calculation can be done in polynomial time. {In the following we show that, among all of the states $(n,\varpi,u_0,u_1,\ldots,u_{\zeta-1})$ that satisfy $f(n,\varpi,u_0,u_1,\ldots,u_{\zeta-1})\le Q$, there exists some state $(n,\varpi^*,u_0^*,u_1^*,\ldots,u_{\zeta-1}^*)$ such that $\varpi^*=|T\cap S_\ell|$ and $u_h^*\le t_h+\epsilon/(mn)$. Denote the set of items selected in the corresponding solution by ${I\cap S_{\ell}}$, then $I\cap S_{\ell}$ satisfies Lemma~\ref{lemma:dp}}.
	
	Consider the optimal solution ${T\cap S_{\ell}}$. Let $T_j={T\cap S_{\ell}}\cap\{1,2,\ldots,j\}$ and $u_i(T_j)=\Pr(P_{T_j}=i)$. We make the following claim.
	
	\begin{claim}\label{claim:1}
		For any $0\le j\le n'$, there exists a vector $(j,\varpi,u_0,u_1,\ldots,u_{\zeta-1})$ such that
		\begin{itemize}
			\item $\varpi= |T_j|$
			\item $f(j,\varpi,u_0,u_1,\ldots,u_{\zeta-1})\le Q_{T_j}$, and
			\item $u_i\le u_i(T_j)+j\eta$ for every $0\le i\le \zeta-1$.
		\end{itemize}
	\end{claim}
	
	\begin{proof}
		%[Proof of {\bf Claim~\ref{claim:1}}] 
		We prove the claim by induction. It is trivial to see that the claim holds when $j=0$. Suppose the claim holds for $j\le s$.
		That is, for $j=s$, there exists some state $(s,\varpi,u_0,u_1,\ldots,u_{\zeta-1})$ such that $u_i\le u_i(T_s)+s\eta$ for every $0\le i\le \zeta-1$ and $f(s,\varpi,u_0,u_1,\ldots,u_{\zeta-1})\le Q_{T_s}$.
		
		Now we prove it holds for $j=s+1$. There are two cases: $s+1\not\in T_{s+1}$ and $s+1\in T_{s+1}$.
		%Consider $T_{s+1}$ in two cases.
		
		In the case $s+1\not\in T_{s+1}$, we have $Q_{T_{s+1}}=Q_{T_s}$ and $u_i(T_{s+1})=u_i(T_s)$. According to Equation~(\ref{eq:recursive-1}), we have
		\begin{eqnarray*}
			&&f(s+1,\varpi,u_0,u_1,\ldots,u_{\zeta-1})\\&\le& f(s,\varpi,u_0,u_1,\ldots,u_{\zeta-1})\le Q_{T_s}=Q_{T_{s+1}},
		\end{eqnarray*}
		hence the claim holds.
		
		In the case $s+1\in T_{s+1}$, we have
		\begin{eqnarray}        
		&&\hspace{-11mm}u_0(T_{s+1})=u_0(T_{s})p_{s+1},\label{eqn:abc-1}\\
		&&\hspace{-11mm} u_i(T_{s+1})=u_i(T_{s})(1-p_{s+1})+u_{i-1}(T_{s})p_{s+1}, 1\le i\le \zeta-1.\label{eqn:abc-2}
		\end{eqnarray}
		Compare Eqs. \eqref{eqn:abc-1}-\eqref{eqn:abc-2} with (\ref{eq:recursive-1}), we see that if
		$$(s+1,\varpi',u_0',u_1',\ldots,u_{\zeta-1}')=(s,\varpi,u_0,u_1,\ldots,u_{\zeta-1})+P_{s+1},$$ then
		$$u_0'\le u_0(1-p_{s+1})+\eta\le u_0(T_{s+1})+(s+1)\eta,$$
		and
		\begin{eqnarray*}
			u_i'&\le& u_i(1-p_{s+1})+u_{i-1}p_{s+1}+\eta\\
			&\le& u_{i}(T_s)(1-p_{s+1})+u_{i-1}(T_{s})p_{s+1}+(s+1)\eta\\&=&u_i(T_{s+1})+(s+1)\eta.
		\end{eqnarray*}
		Furthermore, we have
		\begin{eqnarray*}
			&&f(s+1,\varpi',u_0',u_1',\ldots,u_{\zeta-1}')\\&\le& f(s,\varpi,u_0,u_1,\ldots,u_{\zeta-1})+q_{s+1}\le Q_{T_{s+1}}.
		\end{eqnarray*}
		Hence, {\bf Claim~\ref{claim:1}} holds.
	\end{proof}
	
	%Now we can finish the proof of {\bf Lemma \ref{lemma:dp}}.
	{\bf Claim \ref{claim:1}} says that there exists a state $(n,\varpi^*,u_0^*,u_1^*,\ldots,u_{\zeta-1}^*)$ such that $f(n,\varpi^*,u_0^*,u_1^*,\ldots,u_{\zeta-1}^*)\le Q$, $\varpi^*=|T\cap S_\ell|$ and $u_i^*\le u_i(T_{n'})+n'\eta\le u_i(T\cap S_{\ell})+\epsilon/(mn)$. Since in the recursive calculation we always overestimate (by rounding {\em up}) the probabilities, we have
	$$\Pr(P_{I\cap S_{\ell}}=h)\le u_h^*\le \Pr(P_{T\cap S_{\ell}}=h)+\epsilon/(mn)\le t_h+\epsilon/(mn).$$
	This completes the proof of {\bf Lemma \ref{lemma:dp}}.
\end{proof}

\begin{definition}
	For any subset $D$ of small items and integer $h\ge 0$, we define
	$$\hat{h}_{D}=\frac{h-\sum_{i\in {D}} p_i}{\sqrt{\sum_{i\in {D}}\sigma_i^2}}=\frac{h-\sum_{i\in {D}} p_i}{\sqrt{\sum_{i\in {D}}p_i(1-p_i)}}.$$
\end{definition}

The proofs of the following two lemmas mainly consist of mathematical calculations together with a suitable application of Berry-Essen's theorem.
\begin{lemma}
	\label{lemma:transform}
	If
	$$\sum_{i\in {I\cap S_{\ell}}} p_i>(m/\epsilon)^4 ~~\text{and}~~ |\Phi(\hat{h}_{I\cap S_{\ell}})-\Phi(\hat{h}_{T\cap S_{\ell}})|\le O(\epsilon/m),$$
	then
	$$\Pr\left(\sum_{i\in {I\cap S_{\ell}}}P_i\ge h\right)\ge \Pr\left(\sum_{i\in {T\cap S_{\ell}}}P_i\ge h\right)-\Omega(\epsilon/m),$$
	where $\Phi(x)$ is the cumulative distribution function of the standard normal distribution.
\end{lemma}
\begin{proof}
	We define random variable $X_i=P_i-\EE(P_i)=P_i-p_i$, then $\EE(X_i)=0$,
	\begin{eqnarray*}
		&&\sigma_i^2=\EE(X_i^2)=(1-p_i)^2p_i+p_i^2(1-p_i)=p_i(1-p_i),\\
		&&\rho_i=\EE(|X_i|^3)
		=p_i(1-p_i)[p_i^2+(1-p_i)^2].
	\end{eqnarray*}
	We have
	\begin{eqnarray*}
		\Pr\left(\sum_{i\in {I\cap S_{\ell}}}P_i\ge h \right)&=&\Pr\left(\frac{\sum_{i\in {I\cap S_{\ell}}}X_i}{\sqrt{\sum_{i\in {I\cap S_{\ell}}}\sigma_i^2}}\ge \hat{h}_{I\cap S_{\ell}}\right).
	\end{eqnarray*}    
	According to Berry-Essen's theorem~\eqref{eq:Berry-Essen-theorem}, we have
	$$\left|\Pr\left(\sum_{i\in {I\cap S_{\ell}}}P_i\ge h \right)-\left(1-\Phi(\hat{h}_{I\cap S_{\ell}})\right)\right|
	\le C_0\cdot \frac{\sum_{i\in {I\cap S_{\ell}}}\rho_i}{(\sum_{i\in {I\cap S_{\ell}}}\sigma_i^2)^{3/2}}.$$
	{By plugging in $\rho_i$ and $\sigma_i$}, we have
	\[\Scale[0.9]{$$\left|\Pr\left(\sum_{i\in {I\cap S_{\ell}}}P_i\ge h \right)-\left(1-\Phi(\hat{h}_{I\cap S_{\ell}})\right)\right|
	\le C_0\cdot\frac{1}{\sqrt{\sum_{i\in {I\cap S_{\ell}}}p_i(1-p_i)}}.$$}\]
	For small items, it holds that $1-p_i\ge \epsilon^2$. Since $\sum_{i\in {I\cap S_{\ell}}}p_i>(m/\epsilon)^4$, we have
	$$\left|\Pr\left(\sum_{i\in {I\cap S_{\ell}}}P_i\ge h\right)-\left(1-\Phi(\hat{h}_{I\cap S_{\ell}})\right)\right| \le C_0\epsilon/m.$$
	Hence, the probability $\Pr\left(\sum_{i\in {I\cap S_{\ell}}}P_i\ge h\right)$ can be estimated using the standard normal distribution $\Phi(\hat{h}_{I\cap S_{\ell}})$. More specifically, if we can select $I\cap S_{\ell}$ such that
	$$|\Phi(\hat{h}_{I\cap S_{\ell}})-\Phi(\hat{h}_{T\cap S_{\ell}})|\le O(\epsilon/m)$$
	for every $0\le h\le k$, then it holds that
	\begin{eqnarray*}
		\Pr\left(\sum_{i\in {I\cap S_{\ell}}}P_i\ge h\right)&\ge& \left(1-\Phi(\hat{h}_{I\cap S_{\ell}})\right)- C_0\epsilon/m\\&=&\left(1-\Phi(\hat{h}_{T\cap S_{\ell}})\right)-\Omega(\epsilon/m)\\&\ge& \Pr\left(\sum_{i\in {T\cap S_{\ell}}}P_i\ge h\right)-\Omega(\epsilon/m).
	\end{eqnarray*}
	
	This completes the proof of {\bf Lemma \ref{lemma:transform}}.%\qedhere
\end{proof}

\begin{lemma}\label{lemma:estimate}
	If
	$$\sum_{i\in T\cap S_{\ell}}p_i> (m/\epsilon)^4$$ and the following holds for some $I\cap S_{\ell}$:
	\begin{itemize}
		\item     $|\sum_{i\in {I\cap S_{\ell}}}p_i-\sum_{i\in {T\cap S_{\ell}}}p_i|\le O(\epsilon/m)$, and
		\item $|\sum_{i\in {I\cap S_{\ell}}}p_i(1-p_i)-\sum_{i\in {T\cap S_{\ell}}}p_i(1-p_i)|\le O(\epsilon/m)$,
	\end{itemize}
	then
	$$|\Phi(\hat{h}_{I\cap S_{\ell}})-\Phi(\hat{h}_{T\cap S_{\ell}})|\le O(\epsilon/m)$$ for every $0\le h\le k$.
\end{lemma}

\begin{proof}
	{Note that} $\left|\sum_{i\in {I\cap S_{\ell}}}p_i(1-p_i)-\sum_{i\in {T\cap S_{\ell}}}p_i(1-p_i)\right|\le \epsilon/m$ and  $\sum_{i\in T\cap S_{\ell}}p_i> (m/\epsilon)^4$. We define
	\begin{eqnarray*}
		\beta_1&=&\sqrt{\frac{\sum_{i\in {I\cap S_{\ell}}}p_i(1-p_i)}{\sum_{i\in {T\cap S_{\ell}}}p_i(1-p_i)}}\in [1-\epsilon/m,1+\epsilon/m],\\
		\beta_2&=&\sum_{i\in {I\cap S_{\ell}}}p_i-\sum_{i\in {T\cap S_{\ell}}}p_i\in [-\epsilon/m,\epsilon/m].
	\end{eqnarray*}
	Then, we have
	\begin{eqnarray*}
		\hat{h}_{I\cap S_{\ell}}=\frac{h-\sum_{i\in {I\cap S_{\ell}}}p_i}{\sqrt{\sum_{i\in I\cap S_{\ell}}p_i(1-p_i)}}&=&\frac{1}{\beta_1}\cdot \frac{h-\sum_{i\in {T\cap S_{\ell}}}p_i-\beta_2}{\sqrt{\sum_{i\in T\cap S_{\ell}}p_i(1-p_i)}}\\
		&=& \frac{1}{\beta_1}\cdot \hat{h}_{T\cap S_{\ell}}+O(\epsilon/m).
	\end{eqnarray*}
	Using the result of Claim~\ref{obs:normal} below, we obtain 
	$$|\Phi(\hat{h}_{I\cap S_{\ell}})-\Phi(\hat{h}_{T\cap S_{\ell}})|\le O(\epsilon/m),$$
	meaning Lemma~\ref{lemma:estimate} holds.
	
	\begin{claim}\label{obs:normal}
		For any $x\in(-\infty,\infty)$ and $\delta>0$, it holds that $|\Phi\left((1+\delta)x\pm\delta\right)-\Phi(x)|\le 2\delta$.
	\end{claim}
	
	\begin{proof}
		We observe that for any $y\in (-\infty,\infty)$, it holds that
		$$|\Phi(y\pm\delta)-\Phi(y)|=\left|\int_y^{y\pm\delta}\frac{1}{\sqrt{2\pi}}e^{-t^2/2}dt\right|\le \left|\int_y^{y\pm\delta} 1 \,dt\right|=\delta.$$
		Now we show
		$$|\Phi((1+\delta)x)-\Phi(x)|=\left|\int_{x}^{x+\delta x}\frac{1}{\sqrt{2\pi}}e^{-t^2/2}dt\right|\le \delta.$$
		We observe that because of the symmetry { in $x\in(-\infty,\infty)$}, it suffices to prove the above inequality for $x\ge 0$. In this case, we have $e^{-t^2/2}\le 1/t$ for $t\ge 0$. (This is because the derivative of $te^{-t^2/2}$ is $e^{t^2/2}(1-t^2)$, and consequently its maximum value is $ 1/\sqrt{e}\le 1$.)
		Therefore, we have
		$$\left|\int_{x}^{x+\delta x}\frac{1}{\sqrt{2\pi}}e^{-t^2/2}dt\right|\le \delta x\cdot 1/x=\delta$$
		for $x\ge 0$. Hence, {\bf Claim~\ref{obs:normal}} holds.
	\end{proof}
	
	This completes the proof of {\bf Lemma~\ref{lemma:estimate}}.
\end{proof}

Now we can replace the condition of $\Pr(P_{I\cap S_\ell}\ge h)\ge \Pr(P_{T\cap S_\ell}\ge h)-\Theta(\epsilon/m)$ in {\bf Lemma~\ref{lemma:small}} with the conditions in {\bf Lemma~\ref{lemma:estimate}}, leading to the following lemma whose proof is based on a dynamic programming approach similar to that of {\bf Lemma~\ref{lemma:dp}}.

\begin{lemma}\label{lemma:small-2}
	If $\sum_{i\in T\cap S_{\ell}}p_i> (m/\epsilon)^4$, then there exists an algorithm that runs in $O(m^2n^5/\epsilon^2)$ time and returns a feasible solution {with item set} $I\cap S_{\ell}$ such that
	\begin{itemize}
		\item $Q_{I\cap S_{\ell}}\le Q_{T\cap S_{\ell}}$, and
		\item $|I\cap S_\ell|=|T\cap S_\ell|$, and
		\item     $|\sum_{i\in {I\cap S_{\ell}}}p_i-\sum_{i\in {T\cap S_{\ell}}}p_i|\le O(\epsilon/m)$, and
		\item $|\sum_{i\in {I\cap S_{\ell}}}p_i(1-p_i)-\sum_{i\in {T\cap S_{\ell}}}p_i(1-p_i)|\le O(\epsilon/m)$.
	\end{itemize}
\end{lemma}
%The

\begin{proof}[{\bf Proof of Lemma~\ref{lemma:small}}]
	Without loss of generality we assume $S_\ell=\{1,2,\cdots,n'\}$. Recall that $S_\ell$ consists of small items. For any small item $v_i$ we have $\epsilon^2\le 1-p_i\le 1$.
	We prove {\bf Lemma~\ref{lemma:small}} by considering the following two scenarios:
	\begin{itemize}
		\item {\bf Scenario 1:} $\sum_{i\in T\cap S_{\ell}}p_i\le (m/\epsilon)^4$.
		\item {\bf Scenario 2:} $\sum_{i\in T\cap S_{\ell}}p_i> (m/\epsilon)^4$.
	\end{itemize}
	In {\bf Scenario 1},
	we observe that by Markov's inequality Eq.\eqref{eq:markov-inequality}, we know
	$$\Pr(P_{T\cap S_{\ell}}\ge h)\le \epsilon/m ~~\text{for}~~ h\ge (m/\epsilon)^5.$$
	Let $\zeta=(m/\epsilon)^5$. 
	%Suppose we can find a subset $I\cap S_{\ell}$ of items in polynomial time such that
	{\bf Lemma~\ref{lemma:dp}} showed that we can find a subset $I\cap S_{\ell}$ of items in polynomial time such that
	$$\Pr(P_{I\cap S_{\ell}}=h)\le \Pr(P_{T\cap S_{\ell}}=h)+2\epsilon/(mn)$$
	holds for every $0\le h\le \zeta-1$. Then
	$$\Pr(P_{I\cap S_{\ell}}\ge h)\ge \Pr(P_{T\cap S_{\ell}}\ge h)-2\epsilon/m$$
	for every $0\le h\le \zeta-1$.
	Since
	$$\Pr(P_{I\cap S_{\ell}}\ge h)\ge 0\ge \Pr(P_{T\cap S_{\ell}}\ge h)-2\epsilon/m$$
	for $h\ge \zeta$, we find a near-optimal solution $I\cap S_{\ell}$ in polynomial time. Hence,
	%{What remains to be done is to prove Lemma~\ref{lemma:dp} below.}
	{\bf Lemma~\ref{lemma:small}} holds in {\bf Scenario 1}. 
	%The ``suppose" part, i.e., the existence of a polynomial time algorithm for finding $I\cap S_{\ell}$, is proven in {\bf Lemma~\ref{lemma:dp}}.
	
	In {\bf Scenario 2}, we have $\sum_{i\in T\cap S_{\ell}}p_i> (m/\epsilon)^4$.
	As highlighted before, the difficulty encountered here is to {maximize the probability with respect to the sum of random variables}.
	Our strategy is to first replace the condition
	\begin{eqnarray}\label{eq:condi}
	\Pr(P_{I\cap S_\ell}\ge h)\ge \Pr(P_{T\cap S_\ell}\ge h)-\Theta(\epsilon/m)
	\end{eqnarray}
	in {\bf Lemma~\ref{lemma:small}} with a stronger, but handier condition. More precisely, by {\bf Lemma~\ref{lemma:transform}} and {\bf Lemma~\ref{lemma:estimate}}, we show that Eq~\eqref{eq:condi} is true if we have $|\sum_{i\in {I\cap S_{\ell}}}p_i-\sum_{i\in {T\cap S_{\ell}}}p_i|\le O(\epsilon/m)$, and $|\sum_{i\in {I\cap S_{\ell}}}p_i(1-p_i)-\sum_{i\in {T\cap S_{\ell}}}p_i(1-p_i)|\le O(\epsilon/m)$, plus some cardinality constraints. Note that these knapsack-like constraints are much easier to handle when compared with Eq~\eqref{eq:condi}. We will design a dynamic programming based algorithm that finds a feasible solution with respect to these stronger but handier conditions ({\bf Lemma~\ref{lemma:small-2}}).
	Thus, {\bf Lemma~\ref{lemma:small}} holds in {\bf Scenario 2}.
	%This completes the proof of {\bf Lemma~\ref{lemma:small}}.
\end{proof}

Now we are ready to prove {\bf Theorem \ref{thm:fpt-alg}}.
\begin{proof}[{\bf Proof of Theorem \ref{thm:fpt-alg}}]
	Let $T$ be the set of {\em indices} of items that are selected by the optimal solution to the MKU problem,
	and $\OPT_{MKU}=\Pr\left(\sum_{i\in T} P_i\ge k\right)$ be the optimal objective value given by the optimal solution.
	For any $I$ (i.e., the indices of the items that are selected by an approximation algorithm), we define $P_I=\sum_{i\in I} P_i$ and  $Q_I=\sum_{i\in I}q_i$.
	Denote by $S$ the set of {\em indices} of small items and $B$ the set of {\em indices} of big items. Let $B_j=\{i|i\in B, v_i\in V_j\}$ and $S_j=\{i|i\in S, v_i\in V_j\}$.
	According to the number of big items selected by the optimal solution in each $V_j$, namely $|T\cap B_j|$, we divide the MKU problem into the following two cases:
	\begin{itemize}
		\item {\bf Case 1}: There exists some $1\le j^*\le m$ such that $|T\cap B_{j^*}|\ge 2k$.
		\item {\bf Case 2}: $|T\cap B_{j}|\le 2k-1$ for every $1\le j\le m$.
	\end{itemize}
	In {\bf Case 1}, as $1\le j^*\le m$, we can guess $j^*$
	%, if it exists,
	by $O(m)$ enumerations. { When the guess of $j^*$ is correct, {\bf Theorem \ref{thm:fpt-alg}} is proven as {\bf Lemma~\ref{obs:1}}.}
	
	In {\bf Case 2}, we have $|T\cap B_j|\le 2k-1$ for every $j$. We first guess the values of $|T\cap B_j|$ and $|T\cap S_j|$ for all $j$, leading to $n^{O(m)}$ enumerations. For the correct guess,
	%Suppose we guess out the correct values. Then
	{\bf Corollary~\ref{coro:folklore-sequence}} says that a near optimal solution $I$ can be found when the following conditions are satisfied simultaneously:
	%if it satisfies all the followings:
	\begin{itemize}
		\item $|I\cap B_j|=|T\cap B_j|$ and $|I\cap S_{j}|= T\cap S_j|$.
		\item $Q_{I\cap B_j}\le Q_{T\cap B_j}$ and $Q_{I\cap S_j}\le Q_{T\cap S_j} $.
		\item For $\delta=\Theta(\epsilon/m)$ and any $0\le h\le k$, we have \begin{subequations}
			\begin{eqnarray}
			\Pr(P_{I\cap S_j}\ge h)&\ge& (1-\delta)\Pr(P_{T\cap S_j}\ge h)-\delta,\label{eq:1}\\
			\Pr(P_{I\cap B_j}\ge h)&\ge& (1-\delta)\Pr(P_{T\cap B_j}\ge h)-\delta.\label{eq:2}
			\end{eqnarray}
		\end{subequations}
	\end{itemize}
	This means that we can decompose the MKU problem in {\bf Case 2} into a sequence of sub-problems, each of which asks for a near optimal solution $I\cap B_j$ or $I\cap S_j$. %{\color{red}within or with respect to??} subset $B_j$ or $S_j$.
	Let $1\le \ell\le m$ be an arbitrary index. Then, {\bf Theorem \ref{thm:fpt-alg}} in {\bf Case 2} is proven as {\bf Lemma \ref{lemma:fpt-big}} (dealing with big items) and {\bf Lemma \ref{lemma:small}} (dealing with small items).
	
	%	This completes the proof of {\bf Theorem \ref{thm:fpt-alg}}.
\end{proof}

\section{Conclusion and Discussion}
\label{sec:discussion}

We have introduced the BVU problem with the uncertainty that the vote of a bribed voter may not be counted (either because the bribed voter does not cast its vote in fear of being caught, or because the bribed voter is indeed caught and its vote is discarded).
We have showed that the BVU problem does not admit any {\em multiplicative} $O(1)$-approximation algorithm in FPT time 
{modulo} standard complexity assumptions.
We have also showed that there is an algorithm that returns an approximate solution with an {\em additive}-$\epsilon$ error in FPT time for any arbitrary small $\epsilon$. Given the hardness result, this algorithm is perhaps the best one can hope for.
%We have also discussed a range of open problems for future research.

The BVU problem has many interesting aspects that deserve further studies.
First, our {algorithmic result} assumes a constant number of candidates. Future work needs to characterize the hardness of, and design approximate algorithms for, the BVU problem when the number of candidates is part of the input (rather than a constant). The problem with an arbitrary number of candidates may be strictly harder than that of a constant number of candidates. It is not clear whether or not our approximation algorithm, which works for a constant number of candidates, can be extended to cope with the case of an arbitrary number of candidates.  
Moreover, the hardness of the BVU problem (with both a constant and an arbitrary number of candidates) needs to be investigated with respect to other voting rules, such as the {\em k-approval} or {\em Borda} 
rule. Nevertheless, our hardness result with $m=2$ candidates immediately implies a hardness result 
%with $m=2$ also holds 
with respect to the {\em Borda} voting rule.

Furthermore, the notion of {\em uncertainty} is a rich topic that needs to be explored further.
Even for the particular kind of uncertainty introduced in the present paper, there are many problems that deserve to be studied.
For example, it is interesting to incorporate the {\em probabilistic no-show} introduced by \cite{wojtas2012possible} into our model such that {\em unbribed} voters have some probabilities of no-show (i.e., not casting their votes); of course, the reason that an unbribed voter may not cast its vote is different from the afore-discussed reason that the vote of a bribed voter may not be counted. Another outstanding future work is to consider the probability that a voter accepts a bribe. Moreover, the literature focuses on the setting where the costs of bribery are deterministic. However, such a cost usually is only an estimation because it is private to a voter. Therefore, it is perhaps more reasonable to assume that the probability that a voter takes a bribe depends on the price offered by the briber.

\smallskip

\noindent{\bf Acknowledgement}. We thank the anonymous reviewers for their constructive comments. Research was supported in part by NSF 1756014.

\bibliographystyle{IEEETranN}
\bibliography{ref}

\end{document}